\documentclass[review]{elsarticle}

\usepackage{lineno,hyperref}
\modulolinenumbers[5]

\journal{Systems and Control Letters}

\usepackage{amssymb,latexsym,amsmath}     
\usepackage{graphics} 
\usepackage{graphicx,wrapfig,lipsum}
\usepackage{subcaption}
\usepackage{graphicx}

\usepackage{amsmath,mleftright}
\usepackage{xparse}

\NewDocumentCommand{\evalat}{sO{\big}mm}{%
  \IfBooleanTF{#1}
   {\mleft. #3 \mright|_{#4}}
   {#3#2|_{#4}}%
}

\usepackage{graphics} 
\usepackage{epsfig} 
\usepackage{mathptmx} 
\usepackage{times} 
\usepackage{amsfonts}
\usepackage{amsmath} 
\usepackage{amssymb}  
\usepackage{arydshln}

\newtheorem{theorem}{Theorem}
\newenvironment{proof}{\noindent{\bf Proof:}}{$\hfill \Box$ \vspace{10pt}}

\newtheorem{theorem1}{Theorem1}
\newtheorem{lemma}[theorem1]{Lemma}

\begin{document}

\begin{frontmatter}

\title{Variational dynamic interpolation for kinematic systems on trivial principal bundles}

\author{Sudin Kadam\corref{cor1}}
\ead{sudin@sc.iitb.ac.in}
\author{Ravi N. Banavar\corref{cor2}}
\ead{banavar@iitb.ac.in}
\cortext[cor1]{Corresponding author}
\address{Systems and Control Engineering Department, Indian Institute of Technology Bombay, Mumbai, India, 400072}
    
\begin{abstract}
This article presents the dynamic interpolation problem for locomotion systems evolving on a trivial principal bundle $Q$. Given an ordered set of points in $Q$, we wish to generate a trajectory which passes through these points by 
synthesizing suitable controls. The global product structure of the trivial bundle is used to obtain an induced Riemannian product metric on $Q$. The squared $L^2-$norm of the covariant acceleration is considered as the cost function, and its  first order variations are taken for generating the trajectories. The nonholonomic constraint is enforced through the local form of the principal connection and the group symmetry is employed for reduction. The explicit form of the Riemannian connection for the trivial bundle is employed to arrive at the extremal of the cost function. The result is applied to generate a trajectory for the generalized Purcell's swimmer - a low Reynolds number microswimming mechanism.
\end{abstract}


\end{frontmatter}


\section{Introduction}

Locomotion relates to a variety of movements resulting in transportation from one place to another of biological and robotic systems. Our objective is to generate a trajectory of kinematic locomoting systems which passes through a given set of points. Many of the locomotion systems' configuration space evolves on trivial principal fiber bundle. Some of these are also governed by nonholonomic constraints. Our work is inspired by the tracking problem for systems which are required to pass through certain ordered set of points on the configuration space. This is called as the dynamic interpolation problem and some of the early work on the topic is presented in \cite{jackson1991dynamic}, \cite{lin1982formulation}. The idea of such interpolating curves on non-Euclidean spaces for applications to robotics first appears in \cite{noakes1989cubic}. A variational approach to such problems on Riemannian manifolds and compact connected Lie groups is presented in \cite{crouch1991geometry}. The further extensions to nonholonomic systems using sub-Riemannian geometry is explained in \cite{bloch1993nonholonomic}, \cite{bloch2017variational}, whereas reduction techniques using Lie group symmetries is presented in \cite{altafini2004reduction}.

The present work extends the theory of dynamic interpolation to the principal kinematic form of systems which evolve on trivial principal bundle having a Riemannian structure and admit nonholonomic constraints. Minimization of the squared $L^2-$norm of covariant acceleration is considered as the optimality criterion to obtain the necessary conditions as a differential equation whose solution gives the interpolating curves. As part of the procedure, the expression for the covariant derivative on principal bundle is obtained and symmetry of the structure group is employed to reduce the configuration space which eliminates the group variable. Apart from the objective that the system should pass through the given set of points in the configuration space, the approach presented also gives the expression for the open loop control inputs through the minimizing curve in the base space of the fiber bundle.

\subsection{Organization of the paper}
In the next section we explain the topology of principal bundles that many locomotion systems exhibit and show that this can be treated as a Riemannian product manifold. We also explain the principal kinematic form of systems through the local form of principal connection in this section. Section $3$ explains the preliminaries of Riemannian geometry and also introduces the setting of the problem for the main result of the paper. We then present the dynamic interpolation problem for principal kinematic systems on trivial principal bundles and the ordinary differential equation for the resulting base and group curves is derived. We apply these result to the low Reynolds number generalized Purcell's swimmer in section $4$.

\section{Trivial principal bundles and the Riemannian structure}
The configuration space $Q$ of locomotion systems can usually be written as the product of two manifolds. One is the base manifold $M$ which describes the configuration of the internal shape variables of the mechanism, and the other part is a Lie group $G$ which represents the macro position of the body. Such systems' configuration space admit the topology of a trivial principal fiber bundle\footnote{A \textit{trivial principal fiber bundle} $(Q, \pi, M, G)$ with manifold $M$ as the base space and Lie group $G$ as the structure group is the manifold $Q = M \times G$ and $\pi: Q \to Q/G = M$ as the submersion together with a free left action of $G$ on $Q$ given by left translation in the group variable: $\Phi_h(x,g)=(x, hg)$ for $x \in M$ and $h,g \in G$.} \cite{ostrowski1996geometric}.

\subsection{Riemannian structure on trivial principal bundles}
For a Lie group $G$ with $e$ as its identity, the set of left invariant vector fields on $G$ is isomorphic to its Lie algebra $\mathfrak{g} = T_eG$. By using the left invariance of the Riemannian metric\footnote{A Riemannian metric $\langle \cdot , \cdot \rangle$ on $G$ is said to be left invariant if every left translation is an isometry of $\langle \cdot , \cdot \rangle$.} and defining inner product $\langle \cdot, \cdot \rangle_{\mathfrak{g}}$ on the Lie algebra, we can give $G$ a Riemannian structure as $\langle u , v \rangle_{G} = \langle T_hL_{h^{-1}} u, T_hL_{h^{-1}} v \rangle_{\mathfrak{g}}, \: \forall h \in G \text{ and for } u,v \in T_hG$, \cite{do1992riemannian}.

We consider a trivial principal bundle $(Q, \pi, M, G)$, $x \in M, \: g \in G$ with $\langle \cdot, \cdot \rangle_M$ as the Riemannian metric on $M$ and $\langle \cdot, \cdot \rangle_G$ as the left invariant Riemannian metric on $G$. Using the product structure  $T(Q \times G) \cong TM \times TG \cong TM \times \mathfrak{g}$, $Q$ can be written as a Riemannian product manifold such that for $\dot{x}_1, \dot{x}_2 \in T_xM$ and $\dot{g}_1, \dot{g}_2 \in T_gG$, a metric $\langle \cdot, \cdot \rangle_Q$ on $Q$ is then defined as
\begin{equation}\label{gQ defintion}
\langle (\dot{x}_1, \dot{g}_1),\: (\dot{x}_2, \dot{g}_2) \rangle_{Q} = \langle \dot{x}_1, \dot{x}_2 \rangle_{M} + \langle \dot{g}_1, \dot{g}_2 \rangle_{G} = \langle \dot{x}_1, \dot{x}_2 \rangle_{M} + \langle g_1^{-1}\dot{g}_1, g_2^{-1}\dot{g}_2 \rangle_{\mathfrak{g}}.
\end{equation}
\subsection{Principal kinematic systems}
In our work we consider the principal kinematic form of systems which evolve on trivial principal fiber bundle $Q \ni (x, g)$  and is defined by the kinematic equations given by
\begin{align}
\dot{x} &=u, \label{kinematic_equations1} \\
\xi &= g^{-1}\dot{g} = - \mathbb{A}(x) u  \label{kinematic_equations2}
\end{align}
Here the local connection $\mathbb{A} : TM \mapsto \mathfrak{g}$ maps the points in the tangent bundle $TM$ of the shape space $M$ to the Lie algebra $\mathfrak{g}$ of $G$. We recall that $u \in T_xM$ is the control input which is the shape velocity $\dot{x}\in T_xM$. Also, equations \eqref{kinematic_equations1}, \eqref{kinematic_equations2} together define the nonholonomic constraint $\xi = - \mathbb{A}(x) \dot{x}$. Many robotic and biological locomotion systems are modelled using this paradigm in the kinematic domain \cite{ostrowski1996geometric}, \cite{hatton2013geometric}, \cite{bloch1996nonholonomic}. We also note that although the kinematic form of equations is quite common for many systems, even for dynamical systems, often the simplification in such kinematic form is useful for an analysis for applications such as path planning and trajectory design.

\section{Dynamic interpolation}
In this section we derive the equations of curves as the solution of the dynamic interpolation problem for the principal kinematic form of systems on trivial principle bundles.
\subsection{Preliminaries on the variational principle}
Let M be a smooth Riemannian manifold with the Riemannian metric $g^M$ denoted by $\langle \cdot,\cdot \rangle_{g^M} : T_pM \times T_pM \to \mathbb{R}$ at each point $p \in M$, where $T_pM$ is the tangent space of $M$ at $p$. A Riemannian connection $\overset{M}{\nabla}$ on $M$, is a map that assigns to any two smooth vector fields $X$ and $Y$ on $M$ a new vector field, $\overset{M}{\nabla}_XY$ on $M$. Given vector fields $X$, $Y$ and $Z$ on $M$, the vector field $R_M(X, Y )Z$ is given by 
\begin{equation}\label{RicciTensor}
R_M(X, Y )Z = \overset{M}{\nabla}_X\overset{M}{\nabla}_Y Z - \overset{M}{\nabla}_Y \overset{M}{\nabla}_XZ - \overset{M}{\nabla}_{[X,Y ]}Z,
\end{equation}
where $R_M$ is the curvature tensor on $M$ and $[X, Y ]$ denotes the Lie bracket of the vector fields $X$ and $Y$. For the properties of connection $\overset{M}{\nabla}$ and the curvature, we refer the reader to \cite{do1992riemannian}.

\subsection{Dynamic interpolation for principal kinematic systems}
Consider the set of all $C^1$ piecewise smooth time parameterized curves $q : [T_0, T_N] \to Q=M \times G$ on $(Q, \pi, M, G)$ such that $q(t)=\left(x(t),g(t)\right)$, where $x(t)$ is the curve on $M$ and $g(t)$ is the curve on $G$. Consider proper variations $\underline{q}$ of $q$ from a family of curves passing through fixed points $q_i, \: i = 0,1,...,N$ and fixed times $T_i$ where $T_0 < T_1 < \cdots < T_{N-1} < T_N$, defined by $\underline{q} : (-\epsilon, \epsilon) \times [T_0, T_N ] \to Q, \: (s,t) \mapsto \underline{q}(s,t) = q_s(t)=(x_s(t),g_s(t)), \: s \in (-\epsilon, \epsilon),\: \epsilon \in \mathbb{R}^+$ and 
\begin{align}
& \underline{q}(s,t)|_{s=0} = q(t) = (x(t), g(t) )\qquad \forall t \in [T_0, T_N], \label{condition1} \\
& \underline{q}(s,T_i) = q_i, \qquad \forall s \in (-\epsilon, \epsilon), \quad i = 0,1,...,N. \label{condition2}
\end{align}
We call the curves with fixed $s,\: \underline{q}(s, t)|_{s=const}$ as main curves and those with fixed $t$, $\underline{q}(s, t)|_{t=const}$ as transverse curves. For a fixed $t \in [T_0, T_N]$, the variational vector field $\delta q$ on $Q$ is defined as
\begin{equation}
\delta q(t) = \left( \delta x(t), \delta g(t) \right) = \evalat[\Big]{\frac{d}{ds}}{s=0}\:\underline{q}(s,t).
\end{equation}
The left invariance of the vector fields $\dot{g}$ and $\delta g$ allows us to write them in terms of the pullback to the group identity. We call $\mathfrak{t}(t)$ and $\mathfrak{s}(t)$ as the $\mathfrak{g}$-valued infinitesimal variations corresponding to $\dot{g}(t)$ and $\delta g(t)$, respectively. The tangential and variational vector fields $\dot{q}(t)$ and $\delta q(t)$ can now be written in terms of these $\mathfrak{g}-$valued vector fields using the tangent lift of the group action $T_e\phi: G \times \mathfrak{g} \to \mathfrak{g}$ as
\begin{align}
\dot{q}(t) &= (\dot{x}(t), \dot{g}(t)) = T_e\phi_{g(t)}(\dot{x}(t), \mathfrak{t}(t))= (\dot{x}(t), g(t)\mathfrak{t}(t) ) \text{\qquad \qquad \, (Tangential vector field)} \label{Tangential vector field}\\
\delta q(t) &= (\dot{x}(t), \delta g(t)) = T_e\phi_{g(t)}(\dot{x}(t), \mathfrak{s}(t)) = (\dot{x}(t),g(t)\mathfrak{s}(t)) \text{\qquad \:\:\:\:\:\: (Variational vector field)} \label{Variational vector field}
\end{align}
We note that since the variation is proper, the variational vector field vanishes at all the interpolation points, i.e. $\delta q(T_i) = (\delta x(T_i),\: g^{-1}(T_i)\delta g(T_i)) = (\delta x(T_i),\: \mathfrak{s}(T_i)) = (0, 0)$ for $i = 0, ..., N$. For compactness of notation, in the subsequent parts of the paper we write $q_s(t),\:x_s(t),\:g_s(t)$ as $q_s,\:x_s,\:g_s,$ respectively. We now state the main result of the paper on the dynamic interpolation of principal kinematic systems.
\begin{theorem}
A necessary condition for $q(t) = (x(t), g(t)), \: t \in [T_0, T_N]$ to be an extremal of the cost function $\mathcal{J}(q(t)) = \frac{1}{2} \int_{T_0}^{T_N} \langle \overset{Q}{\nabla}_{\dot{q}(t)} \dot{q}(t), \overset{Q}{\nabla}_{\dot{q}(t)} \dot{q}(t) \rangle_Q dt$ for the principal kinematic system defined by \eqref{kinematic_equations1} and \eqref{kinematic_equations2} over the class of $C^1$ paths $q_s(t)$ on $Q$ satisfying 
\begin{equation}\label{BCs}
q_s(T_i) = q_{T_i} = (x_{T_i},\: g_{T_i})\:\:\: \forall i = 0, ... , N\: \text{ and }\: \dot{q}_s(T_0) = (\dot{x}_{T_0}, \dot{g}_{T_0}),\: \: \dot{q}_s(T_N) = (\dot{x}_{T_N}, \dot{g}_{T_N}) \quad \forall s \in (-\epsilon, \epsilon)
\end{equation}
is that for every $t\in [T_{i-1},T_i], \: i = 1,...N$, the following equations hold -
\begin{align}
& \overset{M}{\nabla^3}_{\dot{x}}\dot{x}+R_M \big( \overset{M}{\nabla}_{\dot{x}}\dot{x},\dot{x}\big) \dot{x} - \mathbb{A}^T(x) \Big( \dddot{\mathfrak{t}}+3\overset{G}{\nabla}_{\mathfrak{t}}\ddot{\mathfrak{t}} + 3 \overset{G}{\nabla}_{\dot{\mathfrak{t}}}\dot{\mathfrak{t}} + \overset{G}{\nabla}_{\ddot{\mathfrak{t}}}\mathfrak{t} + 3\overset{G}{\nabla}\!\!\phantom{.}^{2}_{\mathfrak{t}}\dot{\mathfrak{t}} + 2\overset{G}{\nabla}_{\mathfrak{t}} \overset{G}{\nabla}_{\dot{\mathfrak{t}}}\mathfrak{t} + \nonumber \\ 
& \qquad \qquad \qquad \qquad \qquad \qquad \quad  \overset{G}{\nabla}_{\dot{\mathfrak{t}}} \overset{G}{\nabla}_{\mathfrak{t}} \mathfrak{t} + \overset{G}{\nabla}\!\!\phantom{.}^{3}_{\mathfrak{t}}\mathfrak{t} + R_G(\dot{\mathfrak{t}},\mathfrak{t})\mathfrak{t}+R_G(\overset{G}{\nabla}_\mathfrak{t} \mathfrak{t}, \mathfrak{t})\mathfrak{t} \Big) = 0, \label{FinalEquation} \\
&\mathfrak{t} = - \mathbb{A}(x) \dot{x}. \label{reconstructionEquation}
\end{align}
\end{theorem}
\begin{proof}
We first state a few lemmas which will be used in the derivation of the result.
\begin{lemma}{(Variation on a Riemannian manifold \cite{crouch1991geometry})}\label{Lemma1}
For a smooth curve $x(t) \in M$, $t \in  [T_0, T_N ]$, such that $x(T_i) = x_i, \: \: i = 0, ... N,\: \dot{x}|_{t=T_0} = v_0$ and $\dot{x}|_{t=T_N} = v_N$, $\delta x(t)$ as the variational vector field corresponding to variations $x_s(t)$, the variation of $J(x(t)) = \frac{1}{2} \int_{T_0}^{T_N} \langle \overset{M}{\nabla}_{\dot{x}} \dot{x}, \overset{M}{\nabla}_{\dot{x}} \dot{x} \rangle dt$ is
\begin{align}\label{DynInterpRiemManifold}
\frac{d}{ds} \left( \frac{1}{2} \int_{T_0}^{T_N} \langle \overset{M}{\nabla}_{\dot{x}_s} \dot{x}_s, \overset{M}{\nabla}_{\dot{x}_s} \dot{x}_s \rangle dt \right) = &\int_{T_0}^{T_N} \Big \langle
\overset{G}{\nabla}\!\!\phantom{.}^{3}_{\dot{x}_s} \dot{x}_s + R(\delta x_s, \overset{M}{\nabla}_{\dot{x}_s}\dot{x}_s)\dot{x}_s, \: \overset{M}{\nabla}_{\dot{x}_s} \dot{x}_s \Big \rangle+ \nonumber \\
& \qquad \sum_{i=1}^{N} \left[ \left\langle \overset{M}{\nabla}_{\delta x_s}\dot{x}_s, \overset{M}{\nabla}_{\delta x_s}\dot{x}_s \right\rangle - \Big\langle \overset{M}{\nabla}_{\delta x_s}\dot{x}_s, \delta x_s \Big\rangle \right]_{T_{i-1}^+}^{T_{i}^-}.
\end{align}
\end{lemma}
\begin{lemma}{(Variation on a Lie group \cite{altafini2004reduction})}\label{DynInterpSDPLieGroup}
For a smooth curve $g(t) \in G$, a Lie group, with $\mathfrak{g}-$valued tangent vector field $g^{-1}(t) \dot{g}(t) = \mathfrak{t} \in \mathfrak{g}$ and $\mathfrak{s}$ as the $\mathfrak{g}-$valued variational vector field corresponding to proper variations of $g(t)$. The variation of $J(g_s(t)) = \frac{1}{2} \int_{T_0}^{T_N} \langle \overset{G}{\nabla}_{\dot{g}_s(t)} \dot{g}_s(t), \overset{G}{\nabla}_{\dot{g}_s(t)} \dot{g}_s(t) \rangle dt$ is given as
\begin{align}
\frac{d}{ds} \left( \frac{1}{2} \int_{T_0}^{T_N} \langle \overset{G}{\nabla}_{\dot{g}_s(t)} \dot{g}_s(t), \overset{G}{\nabla}_{\dot{g}_s(t)} \dot{g}_s(t) \rangle dt \right) &= \int_{T_0}^{T_N} \Big \langle
\dddot{\mathfrak{t}}+3\overset{G}{\nabla}_{\mathfrak{t}}\ddot{\mathfrak{t}} + 3 \overset{G}{\nabla}_{\dot{\mathfrak{t}}}\dot{\mathfrak{t}} + \overset{G}{\nabla}_{\ddot{\mathfrak{t}}}\mathfrak{t} + 3\overset{G}{\nabla}\!\!\phantom{.}^{2}_{\mathfrak{t}}\dot{\mathfrak{t}} + 2\overset{G}{\nabla}_{\mathfrak{t}} \overset{G}{\nabla}_{\dot{\mathfrak{t}}}\mathfrak{t} + \overset{G}{\nabla}_{\dot{\mathfrak{t}}} \overset{G}{\nabla}_{\mathfrak{t}} \mathfrak{t} + \nonumber \\ 
& \qquad\qquad \overset{G}{\nabla}\!\!\phantom{.}^{3}_{\mathfrak{t}}\mathfrak{t} + R_G(\dot{\mathfrak{t}},\mathfrak{t})\mathfrak{t}+R_G(\overset{G}{\nabla}_\mathfrak{t} \mathfrak{t}, \mathfrak{t})\mathfrak{t},\: \mathfrak{s} \Big \rangle dt \label{DynInterpLieGroup}
\end{align}
\end{lemma}
We now present the proof of theorem 1. The variation of the cost function $\mathcal{J}$ can be written as
\begin{align*}
\frac{d}{ds} \mathcal{J}(q_s(t)) &= \int_{T_0}^{T_N} \Big\langle \overset{Q}{\nabla}_{\delta q_s} \overset{Q}{\nabla}_{\dot{q}_s} \dot{q}_s, \overset{Q}{\nabla}_{\dot{q}_s} \dot{q}_s \Big\rangle_{Q} dt \\
&= \int_{T_0}^{T_N} \Big\langle \overset{Q}{\nabla}_{(\delta x_s, \delta g_s)} \overset{Q}{\nabla}_{(\dot{x}_s, \dot{g}_s)} (\dot{x}_s, \dot{g}_s), \overset{Q}{\nabla}_{(\dot{x}_s, \dot{g}_s)} (\dot{x}_s, \dot{g}_s) \Big\rangle_{Q} dt \\
&= \int_{T_0}^{T_N} \Big\langle \big( \overset{M}{\nabla}_{\delta x_s} \overset{M}{\nabla}_{\dot{x}_s} \dot{x}_s, \overset{G}{\nabla}_{\delta g_s} \overset{G}{\nabla}_{\dot{g}_s} \dot{g}_s \big), \big( \overset{M}{\nabla}_{\dot{x}_s} \dot{x}_s, \overset{G}{\nabla}_{\dot{g}_s} \dot{g}_s \big) \Big\rangle_{Q} dt \tag*{using the result in \ref{appendixB} }\\
&= \int_{T_0}^{T_N} \left[\Big \langle \overset{M}{\nabla}_{\delta x} \overset{M}{\nabla}_{\dot{x}_s} \dot{x}_s, \overset{M}{\nabla}_{\dot{x}_s} \dot{x}_s \Big\rangle_{M} + \Big\langle \overset{G}{\nabla}_{\delta g_s} \overset{G}{\nabla}_{\dot{g}_s} \dot{g}_s, \overset{G}{\nabla}_{\dot{g}_s} \dot{g}_s \Big \rangle_{G} \right] dt \tag*{from the metric defined in \eqref{gQ defintion} }
\end{align*}
Using lemma 1 and 2 in expanding the first and second terms respectively, and evaluating at $s=0$, we get
\begin{align*}
&\evalat[\Big]{\frac{d}{ds}}{s=0}  \mathcal{J}(q_s(t)) = \int_{T_0}^{T_N} \Bigg [ \Big\langle \overset{M}{\nabla^3}_{\dot{x}}\dot{x}+R_M \Big( \overset{M}{\nabla}_{\dot{x}}\dot{x},\dot{x}\Big) \dot{x}, \delta x_s \Big\rangle_M + \Big\langle \dddot{\mathfrak{t}}+ 3\overset{G}{\nabla}_{\mathfrak{t}}\ddot{\mathfrak{t}} + 3 \overset{G}{\nabla}_{\dot{\mathfrak{t}}}\dot{\mathfrak{t}} + \overset{G}{\nabla}_{\ddot{\mathfrak{t}}}\mathfrak{t}  + 3\overset{G}{\nabla}\!\!\phantom{.}^{2}_{\mathfrak{t}}\dot{\mathfrak{t}} + 2\overset{G}{\nabla}_{\mathfrak{t}} \overset{G}{\nabla}_{\dot{\mathfrak{t}}}\mathfrak{t} + \\
& \quad \overset{G}{\nabla}_{\dot{\mathfrak{t}}} \overset{G}{\nabla}_{\mathfrak{t}} \mathfrak{t} + \overset{G}{\nabla}\!\!\phantom{.}^{3}_{\mathfrak{t}}\mathfrak{t} + R_G(\dot{\mathfrak{t}},\mathfrak{t})\mathfrak{t}+R_G(\overset{G}{\nabla}_\mathfrak{t} \mathfrak{t}, \mathfrak{t})\mathfrak{t} , \mathfrak{s} \Big\rangle_{\mathfrak{g}} \Bigg ]_{s=0} dt + \sum_{i=1}^{N}\left[ \left\langle \overset{M}{\nabla}_{\delta x_s}\dot{x}_s, \overset{M}{\nabla}_{\delta x_s}\dot{x}_s \right\rangle_{s=0} - \Big\langle \overset{M}{\nabla}_{\delta x_s}\dot{x}_s, \delta x_s \Big\rangle_{s=0} \right]_{T_{i-1}^+}^{T_i^-}
\end{align*}
Since the variations are proper, $\delta x_s(t)|_{t=T_i} = 0$ and $\mathfrak{s}|_{t=T_i} = 0$. Hence, we get
\begin{align*}
\evalat[\Big]{\frac{d}{ds}}{s=0}  \mathcal{J}(q_s(t)) &= \int_{T_0}^{T_N} \Bigg [ \left\langle \overset{M}{\nabla^3}_{\dot{x}}\dot{x}+R_M \left( \overset{M}{\nabla}_{\dot{x}}\dot{x},\dot{x}\right) \dot{x}, \delta x_s \right\rangle_M  dt + \Big\langle \dddot{\mathfrak{t}}+3\overset{G}{\nabla}_{\mathfrak{t}}\ddot{\mathfrak{t}} + 3 \overset{G}{\nabla}_{\dot{\mathfrak{t}}}\dot{\mathfrak{t}} + \overset{G}{\nabla}_{\ddot{\mathfrak{t}}}\mathfrak{t} + 3\overset{G}{\nabla}\!\!\phantom{.}^{2}_{\mathfrak{t}}\dot{\mathfrak{t}} +  \\
& \qquad\qquad   2\overset{G}{\nabla}_{\mathfrak{t}} \overset{G}{\nabla}_{\dot{\mathfrak{t}}}\mathfrak{t} + \overset{G}{\nabla}_{\dot{\mathfrak{t}}} \overset{G}{\nabla}_{\mathfrak{t}} \mathfrak{t} + \overset{G}{\nabla}\!\!\phantom{.}^{3}_{\mathfrak{t}}\mathfrak{t} + R_G(\dot{\mathfrak{t}},\mathfrak{t})\mathfrak{t}+R_G(\overset{G}{\nabla}_\mathfrak{t} \mathfrak{t}, \mathfrak{t})\mathfrak{t} , \mathfrak{s} \Big\rangle_{\mathfrak{g}} \Bigg ]_{s=0} dt
\end{align*}
For a principal kinematic system, the nonholonomic constraints are defined by $\xi = \mathfrak{t} = - \mathbb{A}(x) \dot{x}$. To obtain a curve which satisfies these nonholonomic constraints using the variational approach, the variations $\delta x$ and $\mathfrak{s}$ should also satisfy the constraints $\mathfrak{s} = - \mathbb{A}(x) \delta x$. The variational equations thus become
\begin{align}
\evalat[\Big]{\frac{d}{ds}}{s=0}  \mathcal{J}(q_s(t)) &= \int_{T_0}^{T_N}\Bigg [ \left\langle \overset{M}{\nabla^3}_{\dot{x}}\dot{x}+R_M \big( \overset{M}{\nabla}_{\dot{x}}\dot{x},\dot{x}\big) \dot{x}, \delta x_s \right\rangle_M  dt + \Big\langle \dddot{\mathfrak{t}}+3\overset{G}{\nabla}_{\mathfrak{t}}\ddot{\mathfrak{t}} + 3 \overset{G}{\nabla}_{\dot{\mathfrak{t}}}\dot{\mathfrak{t}} + \overset{G}{\nabla}_{\ddot{\mathfrak{t}}}\mathfrak{t} + 3\overset{G}{\nabla}\!\!\phantom{.}^{2}_{\mathfrak{t}}\dot{\mathfrak{t}} + \nonumber \\
& \qquad \qquad \:\:  2\overset{G}{\nabla}_{\mathfrak{t}} \overset{G}{\nabla}_{\dot{\mathfrak{t}}}\mathfrak{t} + \overset{G}{\nabla}_{\dot{\mathfrak{t}}} \overset{G}{\nabla}_{\mathfrak{t}} \mathfrak{t} + \overset{G}{\nabla}\!\!\phantom{.}^{3}_{\mathfrak{t}}\mathfrak{t} + R_G(\dot{\mathfrak{t}},\mathfrak{t})\mathfrak{t}+R_G(\overset{G}{\nabla}_\mathfrak{t} \mathfrak{t}, \mathfrak{t})\mathfrak{t} , - \mathbb{A}(x) \delta x_s \Big\rangle_{\mathfrak{g}}\Bigg ]_{s=0} dt \\
&= \int_{T_0}^{T_N} \Big\langle \overset{M}{\nabla^3}_{\dot{x}}\dot{x}+R_M \big( \overset{M}{\nabla}_{\dot{x}}\dot{x},\dot{x} \big) \dot{x} - \mathbb{A}^T(x) \Big(\dddot{\mathfrak{t}}+3\overset{G}{\nabla}_{\mathfrak{t}}\ddot{\mathfrak{t}} + 3 \overset{G}{\nabla}_{\dot{\mathfrak{t}}}\dot{\mathfrak{t}} + \overset{G}{\nabla}_{\ddot{\mathfrak{t}}} \mathfrak{t} + 3\overset{G}{\nabla}\!\!\phantom{.}^{2}_{\mathfrak{t}}\dot{\mathfrak{t}} + \nonumber \\
& \qquad \qquad 2\overset{G}{\nabla}_{\mathfrak{t}} \overset{G}{\nabla}_{\dot{\mathfrak{t}}}\mathfrak{t} + \overset{G}{\nabla}_{\dot{\mathfrak{t}}} \overset{G}{\nabla}_{\mathfrak{t}} \mathfrak{t} + \overset{G}{\nabla}\!\!\phantom{.}^{3}_{\mathfrak{t}}\mathfrak{t} + R_G(\dot{\mathfrak{t}},\mathfrak{t})\mathfrak{t}+R_G(\overset{G}{\nabla}_\mathfrak{t} \mathfrak{t}, \mathfrak{t})\mathfrak{t}\Big) , \delta x_s \Big\rangle_M dt
\end{align}
Since $\evalat[\Big]{\frac{d}{ds}}{s=0}  \mathcal{J}(q_s(t)) = 0$ for all the admissible variations $\delta q_s = (\delta x_s, \mathfrak{s})$ which satisfy the nonholonomic constraints \eqref{reconstructionEquation}, the result follows.
\end{proof}

\textbf{Remark 3.1:} Consider $n_M$ and $n_G$ as the dimensions of $M$ and $G$, respectively. The equations \eqref{FinalEquation} and \eqref{reconstructionEquation} along with  the group reconstruction equation $\mathfrak{t}(t) = g^{-1}(t)\dot{g}(t)$ give a set of $(n_M+n_G)-$dimensional fourth order equation in $x$ and $g$. These can be solved using the conditions \eqref{BCs} to get the interpolating curve $q(t)$ for the $N$ intervals $[T_i, T_{i+1}]$ using the $4(n_M+n_G)N$ boundary conditions arising from
\begin{itemize}
\item The $(4n_M+4n_G)$ initial and terminal boundary conditions from \eqref{BCs} and
\item The interpolation conditions $q(T_i)^-=q(T_i)^+=q_{T_i}$ along with the $C^2-$smoothness conditions \cite{crouch1991geometry} $\dot{q}(T_i)^-=\dot{q}(T_i)^+, \: \ddot{q}(T_i)^-=\ddot{q}(T_i)^+, \: i = 1,...,N-1$ giving $(4n_M+4n_G)(N-1)$ conditions.
\end{itemize}
We also note that the procedure gives us the value of the control input $u = \dot{x}$ to trace the $C^2$ interpolating trajectory for the system from the given initial condition.
\section{Example : The generalized Purcell's swimmer}
\begin{wrapfigure}{r}{6cm}
\includegraphics[width=6cm]{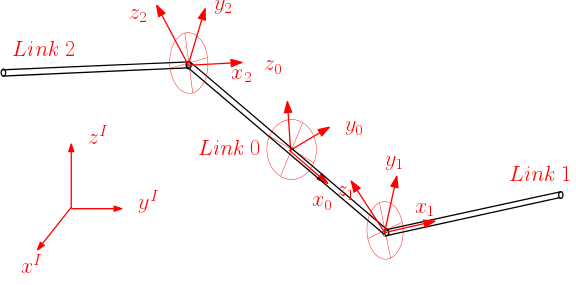}
\caption{The generalized Purcell's swimmer}
\label{Swimmer}
\end{wrapfigure}In this section we apply the dynamic interpolation results to a representative problem of great interest in the control community - the generalized Purcell's swimmer which is a 3-link swimming mechanism moving in a fluid at low Reynolds number conditions \cite{kadam2018geometry}, see Fig.\ref{Swimmer}. The configuration space is a trivial principal bundle with the base space $SO(3) \times SO(3)$ representing the orientation of the coordinate frames $\{x_1,\: y_1,\: z_1 \},\: \{x_2,\: y_2,\: z_2 \} $ associated with the outer links $1$ and $2$ with respect to the base link frame $\{x_0,\: y_0,\: z_0 \}$ through shape variables $R_1, R_2 \in SO(3)$, respectively. The macro-position of the swimmer $g \in SE(3)$ is defined by location of the midpoint of the middle link and its orientation with respect to the inertial frame. We denote by $\xi \in \mathfrak{g}$ the body frame velocity of the base link such that $\xi_R$ and $\xi_T$ are the rotational and translational components of $\xi,$ respectively. $\hat{\omega}_1, \: \hat{\omega}_2 \in \mathfrak{so}(3)$ are the control inputs and represent the shape velocities of the outer links with respect to the base link. The swimmer motion can be written in a principal kinematic form using the Cox theory \cite{kadam2018geometry} with $\mathbb{A}:SO(3)\times SO(3) \times \mathfrak{so}(3) \times \mathfrak{so}(3) \to \mathfrak{se}(3)$.
\subsection{Dynamic interpolation equations for the Purcell's swimmer}
Using the expressions for the covariant derivative and curvature on connected and compact Lie groups\footnote{For a connected compact Lie group if $\nabla$ is the Riemannian connection associated with a left invariant metric and $X, Y$ and $Z$ are left invariant vector fields on the Lie group, then $\nabla_XY = \frac{1}{2}[X.Y]$ and the curvature $R(X,Y) Z = -\frac{1}{4}[[X,Y],Z]$, \cite{milnormorse}.} and adaptation of theorem in \ref{appendixB} to just the base space $M = SO(3) \times SO(3)$, which is a Riemannian product manifold, we get an expression for the first term in \eqref{FinalEquation} for the swimmer as 
\begin{equation}\label{CovDer_Curv_on_base}
\overset{M}{\nabla^3}_{\dot{x}}\dot{x}+R \big( \overset{M}{\nabla}_{\dot{x}}\dot{x},\dot{x}\big) \dot{x} = \left[ \dddot{\omega}_1 - \ddot{\omega}_1 \times \omega_1, \: \dddot{\omega}_2 - \ddot{\omega}_2 \times \omega_2 \right]^T
\end{equation}
Furthermore, for $\xi, \: \eta \in \mathfrak{se}(3)$ with metric tensor $\mathbb{I}$ on $\mathfrak{se}(3)$ as the identity matrix, the relation for covariant derivative of $\xi$ with respect to $•$ is $\eta$ is given as follows, \cite{altafini2004reduction}
\begin{equation}\label{AltafiniCovDerOnLieAlgebra}
\overset{G}\nabla_{\xi}\eta = \frac{1}{2} \left( [\xi,\eta] - \mathbb{I}^{-1} ( ad^*_{\xi} \mathbb{I} \eta + ad^*_{\eta} \mathbb{I} \xi \right) = \begin{bmatrix}
\frac{1}{2} \hat{\omega}_{\xi} \omega_{\eta} \\
\hat{\omega}_{\xi} v_{\eta}
\end{bmatrix}
\end{equation}
For the generalized Purcell's swimmer, using equations \eqref{CovDer_Curv_on_base} and \eqref{AltafiniCovDerOnLieAlgebra}, with $\mathbb{I}$ as the identity, we get the explicit expression for the terms in \eqref{FinalEquation} to get final set of equations as

\begin{align}
& \begin{bmatrix}
\dddot{\omega}_1 - \ddot{\omega}_1 \times \omega_1\\
\dddot{\omega}_2 - \ddot{\omega}_2 \times \omega_2
\end{bmatrix} - \mathbb{A}^T(R_1, R_2) \begin{bmatrix}
\dddot{\xi}_R + \frac{3}{2} \xi_R \times \ddot{\xi}_R + \frac{1}{2} \xi_R \times \xi_R \times \dot{\xi}_R + \frac{1}{2} \xi_R \times \dot{\xi}_R \times \xi_R + ( \dot{\xi}_R \times \xi_R ) \times \xi_R \\
\big( \dddot{\xi}_T + 3 \xi_R \times \ddot{\xi}_T + 3 \dot{\xi}_R \times \dot{\xi}_T + \ddot{\xi}_R \times \xi_T + \frac{5}{2} \xi_R \times \xi_R \times \dot{\xi}_T  + ... \\ 
\qquad ...+\frac{7}{2} \xi_R \times \dot{\xi}_R \times \xi_T + 2 \dot{\xi}_R \times \xi_R \times \xi_T + \frac{1}{2} \xi_R \times \xi_R \times \xi_R \times \xi_T \big)
\end{bmatrix} = 0 \label{Purcell_final_eqn1} \\
& \begin{bmatrix}
\xi_R \\
\xi_T
\end{bmatrix} = -\mathbb{A}(R_1, R_2) \begin{bmatrix}
\omega_1 \\
\omega_2
\end{bmatrix} \label{Purcell_final_eqn2}
\end{align}
Along with the smoothness $C^2$ smoothness conditions and the boundary conditions for $i \in \{0, 1, ... , N \}$ as
\begin{align*}
R_1(T_i)=\bar{R}_{1, T_i}, \:\:\:\: R_2(T_i)=\bar{R}_{2, T_i}, \:\: g(T_i)=\bar{g}_{T_i}, \tag*{(Conditions on the shape and group position)} \\
\omega_1(T_i)=\bar{\omega}_{1, T_i}, \:\:\:\: \omega_2(T_i)=\bar{\omega}_{2, T_i}, \:\:\:\: \xi(T_i)=\bar{\xi}_{T_i} \tag*{(Conditions on shape and group velocities)}
\end{align*}
Lastly, as part of future work, we hope to apply the results in this paper to other locomotion systems and formation control problem consisting of agents whose system model is of the principal kinematic form.
\appendix
\section{}\label{appendixA}
\begin{theorem}\label{Levi-Civita connection on immersed submanifold}
Let $f : M^n \to \overline{M}^{n+k}$ be an immersion of differentiable manifold $M$ into a Riemmanian manifold $\overline{M}$. Assume that $M$ has an induced Riemannian metric. Let for $p \in M, \: U \subset M$ be an open neighbourhood of $p$ such that $f(U) \subset \overline{M}$ is an immersed submanifold of $\overline{M}$. Let $X, Y \in \mathfrak{X} (f(U))$ and $\overline{X}, \overline{Y}$ be their extensions\footnote{Let $\overline{M}$ be a Riemannian manifold and $M$ be its embedded submanifold. Let $X$ be a vector field on $M$, $X \in \mathfrak{X}(M)$. Then $\forall p \in \overline{M}$ and a vector field $\overline{X} \in \mathfrak{X}(\overline{U}$ such that $\overline{X}_p = X_p$/ $\overline{X}$ is called as the (local) extension of $X$.} on $\overline{M}$. We define $\nabla_XY(p) = (\overline{\nabla}_{\overline{X}}\overline{Y})^{T_M}$, where $\overline{\nabla}$ is a Riemannian connection on $\overline{M}$ and $(\overline{\nabla}_{\overline{X}}Y)^{T_M}$ is the tangential component of $\overline{\nabla}_{\overline{X}}Y$ to the manifold $M$. The claim is that $\nabla$ is the Riemannian connection on $M$.
\end{theorem}

\begin{proof}
We show that $\nabla$ satisfies the 2 properties in the Levi-Civita theorem as follows -
\newline\textbf{Symmetry}: For all $p \in M, \: f(p) \in f(M)$ with $x^i$ as the local coordinates of $M$
\begin{align}
(\nabla_XY - \nabla_YX)(p) &= \big( (\overline{\nabla}_{\overline{X}}\overline{Y})^{T_M} - (\overline{\nabla}_{\overline{Y}}\overline{X} )^{T_M} \big)(p) \nonumber \\
&= (\overline{\nabla}_{\overline{X}}\overline{Y} - \overline{\nabla}_{\overline{Y}}\overline{X})^{T_M}(p) \nonumber \\
&= [\overline{X}, \overline{Y}]^{T_M}(p) \nonumber 
\\
&= \Bigg ( \sum_{i,j=1}^{n+k} \left( \overline{X}^i \frac{\partial \overline{Y}^j}{\partial x^i} - \overline{Y}^i \frac{\partial \overline{X}^j}{\partial x^i} \right)\frac{\partial}{\partial x_j} \Bigg )^{T_M}(p) \nonumber \\
&= \Bigg ( \sum_{i=1}^{n}\sum_{j=1}^{n+k} \left( X^i \frac{\partial \overline{Y}^j}{\partial x^i} - Y^i \frac{\partial \overline{X}^j}{\partial x^i} \right)\frac{\partial}{\partial x_j} \Bigg )^{T_M}(p) \tag*{as $X^i,\: Y^i = 0$ for $i = n+1,.., n+k$} \nonumber \\
&= \Bigg ( \sum_{i,j=1}^{n} \left( X^i \frac{\partial Y^j}{\partial x^i} - Y^i \frac{\partial X^j}{\partial x^i} \right)\frac{\partial}{\partial x_j} \Bigg )(p) \nonumber \\
&= [X, Y] \nonumber 
\end{align}
\textbf{Compatibility with the metric on $M$}: 
\begin{align}
X \langle Y, Z \rangle (p) &= \overline{X} \langle \overline{Y}, \overline{Z} \rangle (p) \nonumber \\
&= \langle \overline{\nabla}_{\overline{X}}\overline{Y}, \overline{Z} \rangle(p) + \langle \overline{Y}, \overline{\nabla}_{\overline{X}} \overline{Z} \rangle(p) \nonumber \\
&= \langle (\nabla_X Y)^{T_M}, Z \rangle (p) + \langle Y, (\nabla_X Z)^{T_M} \rangle(p) \nonumber \\
&= \langle \nabla_X Y, Z \rangle (p) + \langle Y, \nabla_X Z \rangle(p) \nonumber 
\end{align}
This shows that $\nabla$ is compatible with the Riemannian metric on  $M$. Hence, using the Levi-Civita theorem \cite{do1992riemannian} $\nabla$ becomes the unique Riemannian connection on $M$ such that $\nabla_XY(p) = (\overline{\nabla}_{\overline{X}}\overline{Y})^{T_M}$.
\end{proof}

\section{}\label{appendixB}
\begin{theorem}\label{Levi Civita connection on trivial principal bundle}
Let $M$ and Lie group $G$ be Riemannian manifolds and consider a trivial principal bundle $(Q, \pi, M ,G)$ as Riemannian product manifold $M \times G$ with the induced product metric $g^M \oplus g^{\mathfrak{g}}$. Let $\overset{M}{\nabla}$ be the Riemannian connection on $M$ and $\overset{G}{\nabla}$ be that on $G$. To prove that the Riemannian connection on $Q$ is given by $\overset{Q}{\nabla}_{(Y_1,Y_2)} (X_1, X_2) = (\overset{M}{\nabla}_{Y_1}X_1, \overset{G}{\nabla}_{Y_2}X_2)$ where $X_1, Y_1 \in \mathfrak{X}(M)$ and $X_2, Y_2 \in \mathfrak{X}(G)$.

\begin{proof}
Let $p \in M, \: q\in G$. Consider $f_M: M \to M \times \{q\}$ and $f_G : G \to \{p \} \times G$ as the canonical maps. We note that $f_1$ and $f_2$ are immersions of $M, \, G$ respectively into $M \times G$. Then using theorem \ref{Levi-Civita connection on immersed submanifold} in \ref{appendixA} and that $\overset{Q}{\nabla}$ is the Levi-Civita connection on $M \times G$, we write that $\overset{M}{\nabla}_{X_1}Y_1(p) = (\overset{Q}{\nabla}_{\overline{X}}\overline{Y})^{T_M}$ and $\overset{G}{\nabla}_{X_2}Y_2(p) = (\overline{\nabla}_{\overline{X}}\overline{Y})^{T_G}$, where $\overline{X}, \, \overline{Y}$ are the extensions of $X_1, X_2$ and $Y_1, Y_2$ defined as follows for all $(a,b) \in M \times G$ -
\begin{equation}
\overline{X}(a,b) = (X_1(a,q), X_2(p,b)), \qquad \overline{Y}(a,b) = (Y_1(a,q), Y_2(p,b))
\end{equation}
Thus, we get
\begin{align}
(\overset{M}{\nabla}_{Y_1}X_1, \overset{G}{\nabla}_{Y_2}X_2) &= ((\overset{Q}{\nabla}_{\overline{Y}}{\overline{X}})^{T_{M}}, (\overset{G}{\nabla}_{\overline{Y}}{\overline{X}})^{T_G})) \\
&= \overset{Q}{\nabla}_{\overline{Y}}{\overline{X}} \\
&= \overset{Q}{\nabla}_{(Y_1, Y_2)}{(X_1,X_2})
\end{align}
\end{proof}
\end{theorem}
\section*{References}
\bibliography{bib_IROS_IFACJSC_DMP_formationControl}

\end{document}